\documentclass[journal]{IEEEtran}
\usepackage{epsf}
\usepackage{graphicx}
\usepackage{epstopdf}
\usepackage{amsfonts,amssymb,amsmath}
\usepackage{algorithmic,algorithm}
\usepackage{subfigure}
\usepackage{latexsym}
\usepackage{bm}
\usepackage[usenames]{color}
\usepackage{url,fancyhdr}
\usepackage{url}

\usepackage[normalem]{ulem}

\newtheorem{theorem}{Theorem}
\newtheorem{lemma}{Lemma}
\newtheorem{corollary}{Corollary}
\newtheorem{Alg}{Algorithm}

\newcommand{\script}[1]{{{\cal{#1} }}}
\newcommand{\expect}[1]{\mathbb{E}\left[#1\right]}

\begin{document}
\title{Power Aware Wireless File Downloading: A Lyapunov Indexing Approach to A Constrained Restless Bandit Problem}
\author{Xiaohan Wei and Michael J. Neely
\thanks{The authors are with the  Electrical Engineering department at the University of Southern California, Los Angeles, CA.}
\thanks{This material was presented in part at the 12th International Symposium on Modeling and
Optimization in Mobile, Ad Hoc, and Wireless  Networks (WiOpt), Hammamet, Tunisia, 2014 \cite{xiaohan-file-download-wiopt2014}.}
\thanks{This
material is supported in part  by one or more of:  the NSF grants CCF-0747525, the Network Science Collaborative Technology Alliance sponsored by the U.S. Army Research Laboratory W911NF-09-2-0053, The Okawa Foundation research award.}}

\markboth{Revised paper for IEEE Transactions on Networking}{}
\maketitle



\begin{abstract}
This paper treats power-aware throughput maximization in a multi-user file downloading system. Each user can receive a new file only after its previous file is finished.  The file state processes for each user act as coupled Markov chains that form a generalized restless bandit system. First, an optimal algorithm is derived for the case of one user. The algorithm maximizes throughput subject to an average power constraint. Next, the one-user algorithm is extended to a low complexity heuristic for the multi-user problem. The heuristic uses a simple online index policy.  In a special case with
no power-constraint, the multi-user heuristic is shown to be throughput optimal.   Simulations are used to demonstrate effectiveness of the heuristic in the general case. For simple cases where the optimal solution can be computed offline, the heuristic is shown to be near-optimal for a wide range of parameters.
\end{abstract}

\section{Introduction}

Consider a wireless access point, such as a base station or  femto node, that delivers files to
$N$ different wireless users.  The system operates in slotted time with time slots $t \in \{0, 1, 2, \ldots\}$.
Each user can download at most one file at a time.   File sizes are random and complete delivery of a file requires a random number of time slots. A new file request is made by each user at a random
time after it finishes its previous download.
Let $F_n(t) \in \{0,1\}$ represent the binary \emph{file state process} for user $n \in \{1, \ldots, N\}$.  The state $F_n(t)=1$ means that user $n$ is currently active downloading a file, while the state $F_n(t)=0$ means that user $n$ is currently idle.

Idle times are assumed to be independent and geometrically distributed with parameter $\lambda_n$ for each user $n$, so that the average idle time is $1/\lambda_n$.  Active times depend on the random file size and the transmission decisions that are made.  Every slot $t$, the access point observes which users are active and
decides to serve a subset of at most $M$ users, where $M$ is the maximum number of simultaneous transmissions allowed in the system ($M < N$ is assumed throughout).
The goal is to maximize a weighted sum of throughput subject to a total average power constraint.

The file state processes $F_n(t)$ are coupled controlled Markov chains that form a total state
$(F_1(t), \ldots, F_N(t))$ that can be viewed as a \emph{restless multi-armed
bandit system}.
Such problems are complex due to the inherent curse of dimensionality.

This paper first computes an online optimal algorithm for 1-user systems, i.e., the case $N=1$.  This simple case avoids the curse of dimensionality and provides valuable intuition.  The optimal policy here is nontrivial and uses the theory of Lyapunov optimization for renewal systems \cite{Neely2010}\cite{renewal-opt-tac}.
 The resulting algorithm makes a greedy transmission decision that affects success probability and power usage. The decision is
 based on a \emph{drift-plus-penalty} index.   Next, the algorithm is extended as a low complexity online heuristic for the $N$-user problem.   The heuristic has the following desirable properties:
\begin{itemize}
\item  Implementation of the $N$-user heuristic is as simple as comparing indices for $N$ different 1-user problems.

\item The $N$-user heuristic is analytically shown to meet the desired average power constraint.

\item  The $N$-user heuristic is shown in simulation to perform well over a wide range of parameters.  Specifically, it is very close to optimal for example cases where an offline optimal can be computed.

\item The $N$-user heuristic is shown to be optimal in a special case with no power constraint and with certain additional assumptions.  The optimality proof uses a theory of \emph{stochastic coupling} for queueing systems \cite{tass-server-allocation}.
\end{itemize}

Prior work on wireless optimization uses Lyapunov functions to maximize throughput in
cases where the users are assumed to have an infinite amount of data to
send \cite{tass-server-allocation}\cite{stolyar-greedy}\cite{now}\cite{atilla-fairness-ton}\cite{neely-fairness-ton}\cite{Srikant2011}\cite{Huang2013}, or
when data arrives according to a fixed rate process that does not depend on delays in the
network  (which necessitates \emph{dropping} data if the arrival rate vector is outside of the capacity
region) \cite{now}\cite{neely-fairness-ton}.  These models do not consider the interplay between arrivals at the transport layer and file delivery at the network layer.    For example, a web user in a coffee shop may want to evaluate the file she downloaded before initiating another download.  The current paper captures this interplay through the binary file state processes $F_n(t)$. This
creates a complex problem of
coupled Markov chains.    This problem is fundamental to
file downloading systems.
The modeling and analysis of these systems is a  significant contribution of the current paper.

To understand this issue, suppose the data arrival rate is fixed and does not adapt to the service received
over the network.   If this arrival rate exceeds network capacity by a factor of two, then at least half of all data must
be dropped.  This can result in an unusable data stream, possibly one that
contains every odd-numbered packet.  A more practical model assumes that full files must be downloaded and that new downloads are only initiated when previous ones are completed.
A general model in this direction would
allow each user to download up to $K$ files simultaneously.
This paper considers the case $K=1$, so that each user is either actively downloading a file, or is idle.\footnote{One way to allow a user $n$ to download up to $K$ files simultaneously is as follows:  Define $K$   \emph{virtual users}  with separate binary file state processes.  The transition probability from idle to active in each of these virtual users is $\lambda_n/K$. The conditional rate of total new arrivals for user $n$ (given that $m$ files are currently in progress) is then $\lambda_n(1-m/K)$ for $m \in \{0, 1, \ldots, M\}$.}
 The resulting
system for $N$ users has a nontrivial Markov structure with $2^N$ states.

Markov decision problems (MDPs) can be solved offline via linear programming \cite{puterman}.  This can be prohibitively complex for large dimensional problems.  Low complexity solutions for coupled MDPs are possible in special cases when the
coupling involves only time average constraints \cite{Neely2012}.  Finite horizon coupled MDPs are treated
via integer programming in \cite{weak1} and via a heuristic ``task decomposition'' method in \cite{weak2}.
The problem of the current paper does not fit the framework of \cite{Neely2012}-\cite{weak2} because it includes both time-average constraints (on average power expenditure) and instantaneous constraints which restrict the number of users that can be served on one slot. The latter service restriction is similar to a traditional restless multi-armed bandit (RMAB) system \cite{whittle}.

{\bf RMAB problems consider a population of $N$ parallel MDPs that continue evolving whether in operation or not (although under different rules). The goal is to choose the MDPs in operation during each time slot so as to maximize the expected reward subject to a constraint on the number of MDPs in operation.} The problem is in general complex (see P-SPACE hardness results in \cite{pspace}).  A standard low-complexity heuristic for such problems is the \emph{Whittle's index} technique \cite{whittle}.
{\bf However, the Whittle's index framework applies only when there are two options on each state (active and passive). Further, it does not consider the
additional time average cost constraints.  The \emph{Lyapunov indexing} algorithm developed in the current paper can
be viewed as an alternative indexing scheme that can always be implemented and that incorporates
additional time average constraints.}
It is likely that the techniques of the current paper can be extended to other constrained
RMAB problems.
Prior work in \cite{tass-server-allocation} develops a Lyapunov drift method for queue stability, and work in
 \cite{Neely2010} develops a drift-plus-penalty ratio method for optimization over renewal systems. The current work is the first to use these techniques as a low complexity heuristic for multidimensional Markov problems.

Work in \cite{tass-server-allocation} uses the theory of stochastic coupling to show that a \emph{longest connected queue} algorithm is delay optimal in a multi-dimensional queueing system with \emph{special symmetric assumptions}.   The problem in \cite{tass-server-allocation} is different from that of the current paper.  However,
a similar coupling approach is used in Section \ref{section:coupling} to show that, for a special case with no power constraint, the Lyapunov indexing
algorithm  is throughput optimal in certain \emph{asymmetric} cases.
 As a consequence, the proof shows the policy is also optimal for a different setting with $M$ servers, $N$ single-buffer queues, and arbitrary packet arrival rates $(\lambda_1, \ldots, \lambda_N)$.

\section{Single user scenario}

Consider a file downloading system that consists of only one user that repeatedly downloads files.
Let $F(t) \in \{0,1\}$ be the file state process of the user.  State ``1'' means there is a file in the system that has not completed its download,
and  ``0'' means no file is waiting.
The length of each file is independent and is either exponentially distributed or geometrically distributed (described in more detail below).  Let $\overline{B}$ denote the expected file size in bits. Time is slotted. At each slot in which there is an active file for downloading, the user makes a service decision that affects both the downloading success probability and the power expenditure. After a file is downloaded, the system goes idle (state $0$) and remains in the idle state for a random amount of time that is independent and geometrically distributed with parameter $\lambda>0$.

A transmission decision is made on each slot $t$ in which $F(t)=1$.  The decision affects the number of bits that are sent, the probability these bits are successfully received, and the power usage.
Let $\alpha(t)$ denote the decision variable at slot $t$ and let $\mathcal{A}$ represent an abstract action set with a finite number of elements.  The set $\script{A}$ can represent a collection of modulation and coding
options for each transmission. Assume also that $\mathcal{A}$ contains an idle action denoted as ``0.'' The
decision  $\alpha(t)$ determines the following two values:
\begin{itemize}
  \item The probability of successfully downloading a file $\phi(\alpha(t))$, where $\phi(\cdot)\in[0,1]$ with $\phi(0)=0$.
  \item The power expenditure $p(\alpha(t))$, where $p(\cdot)$ is a nonnegative function with $p(0)=0$.
\end{itemize}
The user chooses $\alpha(t) = 0$ whenever $F(t)=0$.
The user chooses $\alpha(t) \in \script{A}$ for each slot $t$ in which $F(t) = 1$, with the goal of maximizing throughput subject to a time average power constraint.

The problem can be described by a two state Markov decision process with binary state $F(t)$. Given $F(t)=1$, a file is currently in
the system. This file will finish its download at the end of the slot with probability $\phi(\alpha(t))$. Hence, the transition probabilities out of state $1$ are:
\begin{eqnarray}
Pr[F(t+1) = 0 | F(t)=1] &=& \phi(\alpha(t))  \label{eq:trans1} \\
Pr[F(t+1) = 1 | F(t) = 1] &=& 1-\phi(\alpha(t)) \label{eq:trans2}
\end{eqnarray}
Given $F(t)=0$, the system is idle and will transition to the active state in the next slot with probability $\lambda$,
so that:
\begin{eqnarray}
Pr[F(t+1) = 1 | F(t) = 0] &=& \lambda  \label{eq:trans3} \\
Pr[F(t+1) =0 | F(t) = 0] &=& 1- \lambda \label{eq:trans4}
\end{eqnarray}

Define the throughput, measured by bits per slot, as:
\begin{equation}
    \liminf_{T\rightarrow\infty}\frac{1}{T}\sum_{t=0}^{T-1}\overline{B}\phi(\alpha(t)) \nonumber
\end{equation}
The file downloading problem reduces to the following:
\begin{align}
    \mbox{Maximize:} &~\liminf_{T\rightarrow\infty}\frac{1}{T}\sum_{t=0}^{T-1}\overline{B}\phi(\alpha(t)) \label{eq:prob-1}\\
    \mbox{Subject to:} &~\limsup_{T\rightarrow\infty}\frac{1}{T}\sum_{t=0}^{T-1}p(\alpha(t))\leq \beta \label{eq:prob-2}\\
    & \alpha(t) \in \script{A} \mbox{ $\forall t \in \{0, 1,2, \ldots\}$ such that $F(t) =1$} \label{eq:prob-3} \\
    & \mbox{Transition probabilities satisfy \eqref{eq:trans1}-\eqref{eq:trans4}} \label{eq:prob-4}
\end{align}
where $\beta$ is a positive constant that determines the desired average power constraint.

\subsection{The memoryless file size assumption}

The above model assumes that file completion success on slot $t$ depends only on the
transmission decision $\alpha(t)$, independent of history.   This implicitly assumes that file
length distributions have a \emph{memoryless property} where the residual file length is independent of the amount already delivered. Further, it is assumed that if the controller selects a transmission rate  that is larger than the residual bits in the file, the remaining portion of the transmission is padded with \emph{fill bits}.  This ensures error events provide no information about the residual file length beyond the already known 0/1 binary file state.  Of course, error probability might be improved by removing padded bits. However,  this affects only the last transmission of a file and has negligible impact when expected file size is large in comparison to the amount that can be transmitted in one slot.  Note that padding
is not needed in the special case when all transmissions send one fixed length packet.

The memoryless property holds when each file $i$ has independent length $B_i$ that is  \emph{exponentially distributed} with mean length $\overline{B}$ bits, so that:
\[ Pr[B_i > x] = e^{-x/\overline{B}} \mbox{ for $x >0$} \]
For example, suppose the \emph{transmission rate} $r(t)$ (in units of bits/slot) and the \emph{transmission success probability} $q(t)$ are given by general functions of $\alpha(t)$:
\begin{eqnarray*}
 r(t) &=& \hat{r}(\alpha(t)) \\
 q(t) &=& \hat{q}(\alpha(t))
 \end{eqnarray*}
 Then the file completion probability $\phi(\alpha(t))$ is the probability that the \emph{residual} amount of bits in the file is less than or equal to $r(t)$, \emph{and} that the transmission of these residual bits is a success.  By the memoryless property of the exponential distribution, the residual file length is distributed the same as the original
 file length.  Thus:
 \begin{eqnarray}
\phi(\alpha(t))
  &=& \hat{q}(\alpha(t))Pr[B_i \leq \hat{r}(\alpha(t))] \nonumber \\
  &=& \hat{q}(\alpha(t)) \int_{0}^{\hat{r}(\alpha(t))} \frac{1}{\overline{B}} e^{-x/\overline{B}} dx \label{eq:approx1}
  \end{eqnarray}

Alternatively, history independence holds when each file $i$ consists of a random number $Z_i$ of fixed length packets, where $Z_i$ is geometrically distributed with mean $\overline{Z} = 1/\mu$.   Assume each transmission
sends exactly one packet, but different power levels affect the transmission success probability $q(t) = \hat{q}(\alpha(t))$.   Then:
\begin{equation} \label{eq:approx2}
\phi(\alpha(t)) = \mu\hat{q}(\alpha(t))
\end{equation}

The memoryless file length assumption allows the file state  to be modeled by a simple
binary-valued process $F(t) \in \{0, 1\}$. However, actual file sizes may not have an exponential or geometric distribution.
 One way to treat general distributions is to \emph{approximate} the file sizes as being memoryless by
using a $\phi(\alpha(t))$ function defined by either \eqref{eq:approx1} or \eqref{eq:approx2}, formed by matching the average file size $\overline{B}$ or average number of packets $\overline{Z}$. The \emph{decisions} $\alpha(t)$ are made according to the algorithm below, but the actual event outcomes that arise from these decisions are not memoryless.   A simulation comparison of this approximation is provided in Section \ref{section:sims}, where it is shown to be remarkably accurate (see Fig. \ref{fig:Stupendous6}).

 The algorithm in this section optimizes over the class of all algorithms that do not use residual file length information.
This maintains low complexity by ensuring a user has a binary-valued
Markov state $F(t) \in \{0,1\}$. While a system controller might know the
residual file length, incorporating this knowledge creates a Markov decision problem with an infinite number of states (one for each possible value of residual length) which significantly complicates the scenario.

\subsection{Lyapunov optimization}

This subsection develops an online algorithm for problem \eqref{eq:prob-1}-\eqref{eq:prob-4}. First, notice that file state ``$1$'' is recurrent under any decisions for $\alpha(t)$. Denote $t_k$ as the $k$-th time when the system returns to state ``1.''
Define the \emph{renewal frame} as the time period between $t_k$ and $t_{k+1}$.  Define the \emph{frame size}:
\[ T[k] = t_{k+1} - t_k \]
Notice that $T[k]=1$ for any frame $k$ in which the file does not complete its download.  If the file is completed on frame $k$, then $T[k] = 1 + G_k$, where $G_k$ is a geometric random variable with mean
$\expect{G_k} = 1/\lambda$.   Each frame $k$ involves only a single decision $\alpha(t_k)$ that is made at the beginning of the frame.  Thus, the total power used over the duration of frame $k$ is:
\begin{equation}\label{extra1}
\sum_{t=t_{k}}^{t_{k+1}-1}p(\alpha(t)) = p(\alpha(t_k))
\end{equation}
Using a technique similar to that
proposed in \cite{Neely2010}, we treat the time average constraint in \eqref{eq:prob-2} using a virtual queue $Q[k]$ that is updated every frame $k$ by:
\begin{equation}
    Q[k+1]=\max\left\{Q[k]
     + p(\alpha(t_k)) - \beta T[k],~0\right\} \label{eq:q-update}
\end{equation}
with initial condition $Q[0]=0$.
The algorithm is then parameterized by a constant $V\geq0$ which affects a performance tradeoff.  At the beginning of the $k$-th renewal frame, the user observes virtual queue $Q[k]$ and chooses $\alpha(t_k)$ to maximize the following drift-plus-penalty (DPP) ratio \cite{Neely2010}:
\begin{equation}\label{e4}
    \max_{\alpha(t_k)\in\mathcal{A}} ~~\frac{V\overline{B}\phi(\alpha(t_k))- Q[k]p(\alpha(t_k))}
        {\mathbb{E}[T[k]|\alpha(t_k)]}
\end{equation}
The numerator of the above ratio adds a ``queue drift term''  $-Q[k]p(\alpha(t_k))$ to the
``current reward term'' $V\overline{B}\phi(\alpha(t_k))$. The intuition
is that  it is desirable to have a large value of current reward, but it is also desirable  to have a large drift (since this tends to decrease queue size).   Creating a weighted sum of these two terms and dividing by the expected frame size
gives a simple index.  The next subsections show that, for the context of the current paper, this index leads to an algorithm that pushes throughput arbitrarily close to optimal (depending on the chosen $V$ parameter) with a
strong sample path guarantee on average power expenditure.

The denominator in \eqref{e4} can  easily be computed:
\begin{eqnarray*}
\mathbb{E}[T[k]|\alpha(t_k)] = 1+\frac{\phi(\alpha(t_k))}{\lambda} \nonumber
\end{eqnarray*}
Thus, \eqref{e4} is equivalent to
\begin{equation}\label{e5}
    \max_{\alpha(t_k)\in\mathcal{A}} ~~\frac{V\overline{B}\phi(\alpha(t_k))- Q[k]p(\alpha(t_k))}
        {1+\phi(\alpha(t_k))/\lambda}
\end{equation}
Since there are only a finite number of elements in $\mathcal{A}$, (\ref{e5}) is easily computed.  This gives the following algorithm for the single-user case:
\begin{algorithm}
\begin{Alg}~
\begin{itemize}
  \item At each time $t_{k}$, the user observes virtual queue $Q[k]$ and chooses $\alpha(t_k)$ as the solution to (\ref{e5}) (where ties are broken arbitrarily).
  \item The value $Q[k+1]$ is computed according to \eqref{eq:q-update} at the end of the $k$-th frame.
\end{itemize}
\end{Alg}
\end{algorithm}

\subsection{Average power constraints via queue bounds}

\begin{lemma} \label{lem:1}
If there is a constant $C\geq 0$ such that $Q[k] \leq C$ for all $k \in \{0, 1, 2, \ldots\}$, then:
\[ \limsup_{T\rightarrow\infty}\frac{1}{T}\sum_{t=0}^{T-1}p(\alpha(t))\leq \beta \]
\end{lemma}
\begin{proof}
From \eqref{eq:q-update}, we know that for each frame $k$:
\begin{equation}
    Q[k+1] \geq Q[k]
     +p(\alpha(t_{k}))-T[k]\beta \nonumber
\end{equation}
Rearranging terms and using $T[k] = t_{k+1} -t_k$ gives:
\[ p(\alpha(t_k)) \leq (t_{k+1} - t_k)\beta + Q[k+1]-Q[k] \]
Fix $K>0$. Summing over $k\in\{0,1,\cdots,K-1\}$  gives:
\begin{eqnarray*}
    \sum_{k=0}^{K-1}p(\alpha(t_{k})) &\leq& (t_{K}-t_0) \beta + Q[K] - Q[0] \\
    &\leq&  t_K\beta + C
\end{eqnarray*}
The sum power over the first $K$ frames is the same as the sum up to time $t_{K}-1$, and so:
\[ \sum_{t=0}^{t_K-1} p(\alpha(t)) \leq t_K \beta + C \]
Dividing by $t_K$ gives:
\[ \frac{1}{t_K}\sum_{t=0}^{t_K-1} p(\alpha(t)) \leq \beta + C/t_K. \]
Taking $K\rightarrow\infty$, then,
\begin{equation} \label{eq:sub-here}
\limsup_{K\rightarrow\infty}\frac{1}{t_K}\sum_{t=0}^{t_K-1} p(\alpha(t)) \leq \beta
\end{equation}
Now for each positive integer $T$, let $K(T)$ be the integer such that $t_{K(T)} \leq T < t_{K(T)+1}$.
Since power
is only used at the first slot of a  frame, one has:
\[ \frac{1}{T}\sum_{t=0}^{T-1} p(\alpha(t))\leq \frac{1}{t_{K(T)}}\sum_{t=0}^{t_{K(T)}-1}p(\alpha(t)) \]
Taking a $\limsup$ as $T\rightarrow\infty$ and using \eqref{eq:sub-here} yields the result.
\end{proof}

The next lemma shows that the queue process under our proposed algorithm is deterministically bounded.
Define:
\begin{eqnarray*}
p^{min} &=& \min_{\alpha\in\mathcal{A}\setminus\{0\}}p(\alpha)\\
p^{max}&=&\max_{\alpha\in\mathcal{A}\setminus\{0\}}p(\alpha)
\end{eqnarray*}
Assume that $p^{min}>0$.

\begin{lemma} \label{lem:2} If $Q[0]=0$, then under our algorithm we have for all $k>0$:
\[ Q[k]\leq \max\left\{\frac{V\overline{B}}{p^{min}}+p^{max}-\beta,0\right\}\]
\end{lemma}

\begin{proof}
First, consider the case when $p^{max} \leq \beta$. From \eqref{eq:q-update} and the fact that $T[k] \geq 1$ for all $k$,
 it is clear  the queue can never increase, and so $Q[k] \leq Q[0]=0$ for all $k>0$.

Next, consider the case when $p^{max} > \beta$. We prove the assertion by induction on $k$.
The result trivially holds for $k=0$.  Suppose it holds at $k=l$ for $l>0$, so that:
\[ Q[l]\leq \frac{V\overline{B}}{p^{min}}+p^{max}-\beta \]
We are going to prove that the same holds for $k=l+1$.  There are two cases:
\begin{enumerate}
  \item $Q[l]\leq\frac{V\overline{B}}{p^{min}}$. In this case we have by \eqref{eq:q-update}:
  \begin{eqnarray*}
   Q[l+1] &\leq& Q[l] + p^{max} - \beta \\
   &\leq& \frac{V\overline{B}}{p^{min}} + p^{max} - \beta
   \end{eqnarray*}

    \item $\frac{V\overline{B}}{p^{min}}<Q[l]\leq\frac{V\overline{B}}{p^{min}}+p^{max}-\beta$.    In this case,
    if $p(\alpha(t_l))=0$ then the queue cannot increase, so:
    \[ Q[l+1] \leq Q[l] \leq   \frac{V\overline{B}}{p^{min}} + p^{max} - \beta  \]
    On the other hand, if $p(\alpha(t_l))>0$ then $p(\alpha(t_l)) \geq p^{min}$ and so the numerator in
    \eqref{e5} satisfies:
    \begin{eqnarray*}
    V\overline{B}\phi(\alpha(t_l)) - Q[l]p(\alpha(t_l))
    &\leq& V\overline{B} - Q[l]p^{min} \\
    &<& 0
    \end{eqnarray*}
    and so the maximizing ratio in \eqref{e5} is negative.  However, the maximizing ratio in \eqref{e5} \emph{cannot} be negative
    because the alternative  choice $\alpha(t_l)=0$ increases the ratio to 0.  This contradiction implies that
    we cannot have $p(\alpha(t_l))>0$.
\end{enumerate}
\end{proof}

The above is a \emph{sample path result} that only assumes parameters satisfy $\lambda >0$, $\overline{B}>0$, and $0 \leq \phi(\cdot) \leq 1$.  Thus, the algorithm meets the average power constraint even if it uses incorrect values for these parameters. The next subsection provides a throughput optimality result when these parameters match the true system values.

\subsection{Optimality over randomized algorithms}

Consider the following class of \emph{i.i.d. randomized algorithms}:  Let $\theta(\alpha)$ be non-negative numbers defined for each $\alpha \in \script{A}$, and suppose they satisfy $\sum_{\alpha \in \script{A}} \theta(\alpha) = 1$.  Let $\alpha^*(t)$ represent a policy that, every slot $t$ for which $F(t)=1$, chooses $\alpha^*(t) \in \script{A}$ by independently selecting strategy $\alpha$ with probability $\theta(\alpha)$.    Then $(p(\alpha^*(t_k)), \phi(\alpha^*(t_k)))$ are independent and identically distributed (i.i.d.) over frames $k$. Under this algorithm, it follows by the law of large numbers that the throughput and power expenditure satisfy (with probability 1):
\begin{eqnarray*}
\lim_{t\rightarrow\infty} \frac{1}{T}\sum_{t=0}^{T-1} \overline{B}\phi(\alpha^*(t)) &=& \frac{\overline{B}\expect{\phi(\alpha^*(t_k))}}{1 + \expect{\phi(\alpha^*(t_k))}/\lambda} \\
\lim_{t\rightarrow\infty} \frac{1}{T}\sum_{t=0}^{T-1} p(\alpha^*(t)) &=& \frac{\expect{p(\alpha^*(t_k))}}{1 + \expect{\phi(\alpha^*(t_k))}/\lambda}
\end{eqnarray*}
It can be shown that optimality of problem \eqref{eq:prob-1}-\eqref{eq:prob-4} can be achieved over this class \cite{renewal-opt-tac}. Thus, there exists an i.i.d. randomized algorithm $\alpha^*(t)$ that satisfies:
\begin{eqnarray}
\frac{\overline{B}\expect{\phi(\alpha^*(t_k))}}{1 + \expect{\phi(\alpha^*(t_k))}/\lambda} &=& \mu^* \label{eq:iid1} \\
 \frac{\expect{p(\alpha^*(t_k))}}{1 + \expect{\phi(\alpha^*(t_k))}/\lambda} &\leq& \beta \label{eq:iid2}
\end{eqnarray}
where $\mu^*$ is the optimal throughput for the problem \eqref{eq:prob-1}-\eqref{eq:prob-4}.

\subsection{Key feature of the drift-plus-penalty ratio}

Define $\script{H}[k]$ as the \emph{system history} up to frame $k$, which includes the actions taken $\alpha[0],\cdots,\alpha[k-1]$ frame lengths $T[0],\cdots,T[k-1]$, the busy period in each frame, the idle period in each frame, and the queue value $Q[k]$ (since
this is determined by the random events before frame $k$).
Consider the algorithm
that, on frame $k$, observes $Q[k]$ and
chooses $\alpha(t_k)$ according to \eqref{e5}.  The following key feature of this algorithm can be
shown (see \cite{Neely2010} for related results):
\begin{align*}
&\frac{\expect{-V\overline{B}\phi(\alpha(t_k)) + Q[k]p(\alpha(t_k))|\script{H}[k]}}{\expect{1 + \phi(\alpha(t_k))/\lambda|\script{H}[k]}}
\nonumber \\
&\leq
\frac{\expect{-V\overline{B}\phi(\alpha^*(t_k)) + Q[k]p(\alpha^*(t_k))|\script{H}[k]}}{\expect{1 + \phi(\alpha^*(t_k))/\lambda|\script{H}[k]}}
\end{align*}
where $\alpha^*(t_k)$ is any (possibly randomized) alternative decision that is based only on $\script{H}[k]$.
 This is an intuitive property:  By design, the algorithm in \eqref{e5} observes $\script{H}[k]$ and then chooses a particular action $\alpha(t_k)$ to minimize the ratio over all deterministic actions.  Thus, as can be shown, it also minimizes the ratio over all potentially randomized actions.
Using the (randomized) i.i.d. decision $\alpha^*(t_k)$ from \eqref{eq:iid1}-\eqref{eq:iid2} in the above and noting that this alternative decision
is independent of $\script{H}[k]$ gives:
\begin{align}
&\frac{\expect{-V\overline{B}\phi(\alpha(t_k)) + Q[k]p(\alpha(t_k))|\script{H}[k]}}{\expect{1 + \phi(\alpha(t_k))/\lambda|\script{H}[k]}}
\leq -V\mu^* + Q[k]\beta \label{eq:key-feature}
\end{align}

\subsection{Performance theorem}

\begin{theorem}
The proposed algorithm achieves the constraint $\limsup_{T\rightarrow\infty}\frac{1}{T}\sum_{t=0}^{T-1}p(\alpha(t))\leq \beta$ and yields throughput satisfying (with probability 1):
\begin{equation}\label{e6}
    \liminf_{T\rightarrow\infty}\frac{1}{T}\sum_{t=0}^{T-1}\overline{B}\phi(\alpha(t))\geq \mu^{*}-\frac{C_0}{V}
\end{equation}
where $C_0$ is a constant.\footnote{The constant $C_0$ is independent of $V$ and is given in the proof.}
\end{theorem}
\begin{proof}
First, for any fixed $V$, Lemma \ref{lem:2} implies that the queue is deterministically bounded.  Thus, according to Lemma \ref{lem:1}, the proposed algorithm achieves the constraint $\limsup_{T\rightarrow\infty}\frac{1}{T}\sum_{t=0}^{T-1}p(\alpha(t))\leq \beta$. The rest is devoted to proving the throughput guarantee \eqref{e6}.

Define:
\begin{equation}
    L(Q[k])=\frac{1}{2}Q[k]^2. \nonumber
\end{equation}
We call this a \emph{Lyapunov function}. Define a frame-based Lyapunov Drift as:
\[ \Delta[k] = L(Q[k+1]) - L(Q[k]) \]
According to \eqref{eq:q-update} we get
\begin{equation}
    Q[k+1]^2\leq \left(Q[k]+p(\alpha(t_{k}))-T[k]\beta\right)^2. \nonumber
\end{equation}
Thus:
\begin{eqnarray*}
\Delta[k]  \leq \frac{(p(\alpha(t_k)) - T[k]\beta)^2}{2} + Q[k](p(\alpha(t_k)) - T[k]\beta)
\end{eqnarray*}
Taking a conditional expectation of the above given $\script{H}[k]$ and recalling that $\script{H}[k]$ includes the information
$Q[k]$ gives:
\begin{equation} \label{eq:drift}
 \expect{\Delta[k] |\script{H}[k]} \leq C_0 + Q[k]\expect{p(\alpha(t_k)) - \beta T[k]| \script{H}[k]}
 \end{equation}
where $C_0$ is a constant that satisfies the following for all possible histories $\script{H}[k]$:
\[ \expect{\frac{(p(\alpha(t_k)) - T[k]\beta)^2}{2} \left |\right. \script{H}[k]} \leq C_0 \]
Such a constant $C_0$ exists because the power $p(\alpha(t_k))$ is deterministically bounded, and the frame sizes $T[k]$ are bounded in second moment regardless of history.

Adding the ``penalty'' $-\expect{V\overline{B}\phi(\alpha(t_k))|\script{H}[k]}$ to both sides of \eqref{eq:drift} gives:
\begin{align*}
&\expect{\Delta[k] - V\overline{B}\phi(\alpha(t_k))|\script{H}[k]}  \\
&\leq C_0 +  \expect{-V\overline{B}\phi(\alpha(t_k)) + Q[k](p(\alpha(t_k)) - \beta T[k])| \script{H}[k]} \\
&= C_0 - Q[k]\beta\expect{T[k]|\script{H}[k]}  \\
& +  \frac{\expect{T[k]|\script{H}[k]}\expect{-V\overline{B}\phi(\alpha(t_k)) + Q[k]p(\alpha(t_k))|\script{H}[k]}}{\expect{T[k]|\script{H}[k]}}
\end{align*}
Expanding $T[k]$ in the denominator of the last term  gives:
\begin{align*}
&\expect{\Delta[k] - V\overline{B}\phi(\alpha(t_k))|\script{H}[k]}  \nonumber\\
&\leq C_0 - Q[k]\beta\expect{T[k]|\script{H}[k]}  + \expect{T[k]|\script{H}[k]} \times \nonumber \\
&\frac{\expect{-V\overline{B}\phi(\alpha(t_k)) + Q[k]p(\alpha(t_k))|\script{H}[k]}}{\expect{1 + \phi(\alpha(t_k))/\lambda|\script{H}[k]}}
\end{align*}
Substituting \eqref{eq:key-feature} into the above expression gives:
\begin{eqnarray}
&&\hspace{-.4in}\expect{\Delta[k] - V\overline{B}\phi(\alpha(t_k))|\script{H}[k]} \nonumber  \\
&\leq&  C_0 - Q[k]\beta\expect{T[k]|\script{H}[k]} \nonumber  \\
&& + \expect{T[k]|\script{H}[k]} (-V\mu^* + \beta Q[k])\nonumber \\
&=& C_0 -V\mu^* \expect{T[k] | \script{H}[k]}  \label{eq:dude}
\end{eqnarray}
Rearranging gives:
\begin{equation} \label{eq:dpp}
\expect{\Delta[k] + V(\mu^*T[k] - \overline{B}\phi(\alpha(t_k)))|\script{H}[k]} \leq C_0 \
\end{equation}

The above is a drift-plus-penalty expression.
Because we already know the queue $Q[k]$ is deterministically bounded, it follows that:
\[ \sum_{k=1}^{\infty} \frac{\expect{\Delta[k]^2}}{k^2} < \infty \]
This, together with \eqref{eq:dpp}, implies by the drift-plus-penalty result in
Proposition 2 of \cite{JAM2012}  that (with probability 1):
\[  \limsup_{K\rightarrow\infty} \frac{1}{K}\sum_{k=0}^{K-1} \left[\mu^* T[k] - \overline{B}\phi(\alpha(t_k))\right] \leq \frac{C_0}{V} \]
Thus, for any $\epsilon>0$ one has for all sufficiently large $K$:
\[ \frac{1}{K}\sum_{k=0}^{K-1}[\mu^* T[k] - \overline{B} \phi(\alpha(t_k))] \leq \frac{C_0}{V} + \epsilon \]
Rearranging implies that for all sufficiently large $K$:
\begin{eqnarray*}
\frac{\sum_{k=0}^{K-1} \overline{B}\phi(\alpha(t_k))}{\sum_{k=0}^{K-1} T[k]} &\geq& \mu^*  - \frac{(C_0/V + \epsilon)}{\frac{1}{K}\sum_{k=0}^{K-1}T[k]}\\
&\geq& \mu^* - (C_0/V + \epsilon)
\end{eqnarray*}
where the final inequality holds because $T[k] \geq 1$ for all $k$. Thus:
\[ \liminf_{K\rightarrow\infty} \frac{\sum_{k=0}^{K-1} \overline{B}\phi(\alpha(t_k))}{\sum_{k=0}^{K-1} T[k]} \geq \mu^* - (C_0/V + \epsilon) \]
The above holds for all $\epsilon>0$.  Taking a limit as $\epsilon\rightarrow 0$ implies:
\[ \liminf_{K\rightarrow\infty} \frac{\sum_{k=0}^{K-1} \overline{B}\phi(\alpha(t_k))}{\sum_{k=0}^{K-1} T[k]} \geq \mu^* - C_0/V. \]
Notice that $\phi(\alpha(t))$ only changes at the boundary of each frame and remains 0 within the frame. Thus, we can replace the sum over frames $k$ by a sum over slots $t$. The desired result follows.
\end{proof}

The theorem shows that throughput can be pushed within $O(1/V)$ of the optimal value $\mu^*$, where $V$ can be chosen as large as desired to ensure throughput is arbitrarily close to optimal.  The tradeoff is a queue bound that grows linearly with $V$ according to Lemma \ref{lem:2}, which affects the convergence time required for the constraints to be close to the desired time averages (as described in the proof of Lemma \ref{lem:1}).

\section{Multi-user file downloading}

\begin{figure}[htbp]
   \centering
   \includegraphics[height=2.5in]{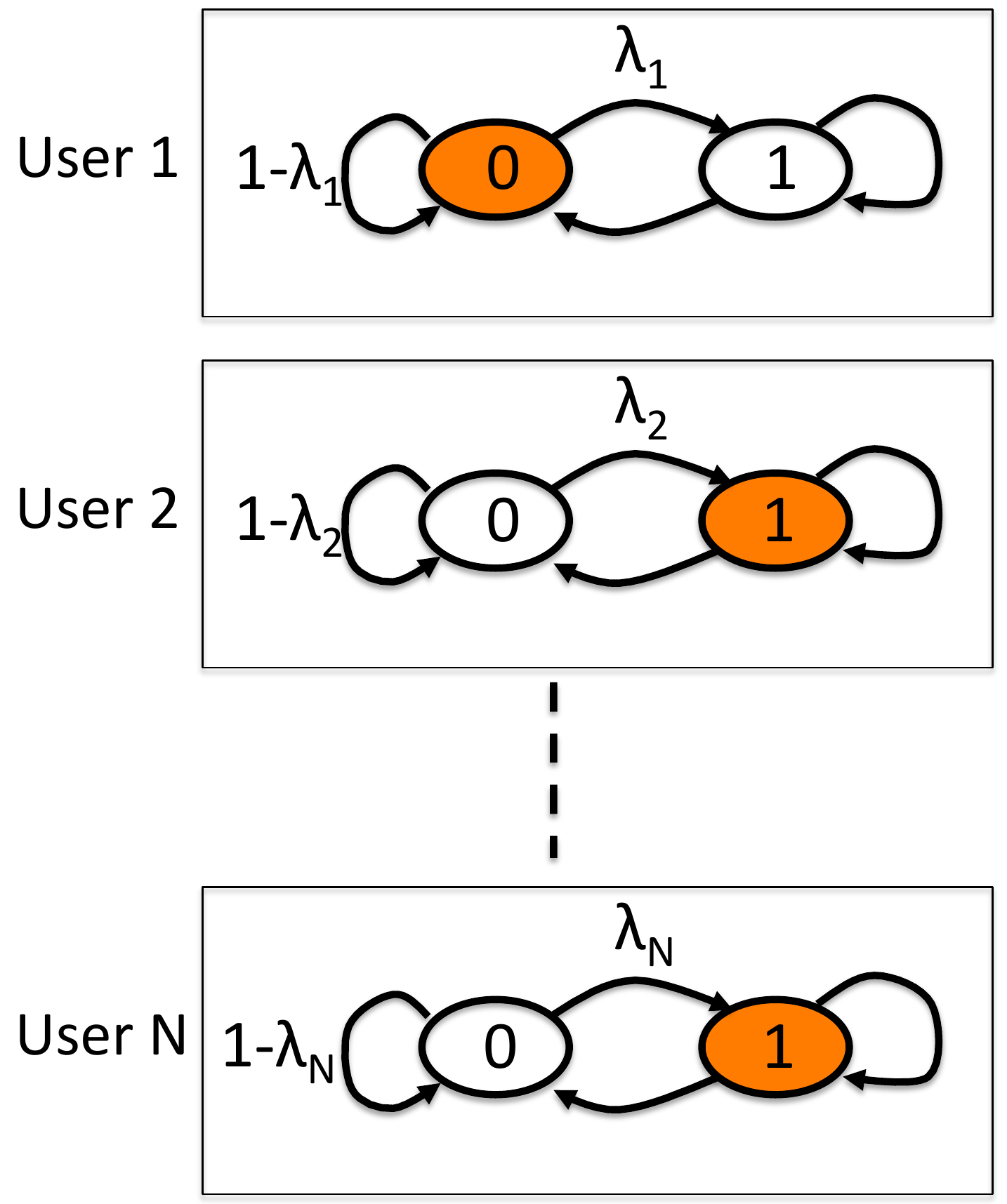} 
   \caption{A system with $N$ users.  The shaded node for each user $n$ indicates the current file state $F_n(t)$  of that user.  There are $2^N$ different state vectors.}
   \label{fig:multi-user-picture}
\end{figure}

This section considers a multi-user file downloading system that
consists of $N$ single-user subsystems.  Each subsystem is similar to the single-user system described in the previous section.
Specifically, for  the $n$-th user (where $n \in \{1, \ldots, N\}$):
\begin{itemize}
\item The file state process is $F_n(t) \in \{0,1\}$.
\item The transmission decision is $\alpha_n(t) \in \script{A}_n$, where $\script{A}_n$ is an abstract set of transmission options for user $n$.
\item The power expenditure on slot $t$ is $p_n(\alpha_n(t))$.
\item The success probability on a slot $t$ for which $F_n(t)=1$ is $\phi_n(\alpha_n(t))$, where $\phi_n(\cdot)$ is the function that describes file completion probability for user $n$.
\item The idle period parameter is $\lambda_n>0$.
\item The average file size is $\overline{B}_n$ bits.
\end{itemize}
Assume that the random variables associated with different subsystems are mutually independent.  The resulting Markov decision problem has $2^N$ states, as shown in Fig. \ref{fig:multi-user-picture}. The transition probabilities for each active user depends on which users are selected for transmission and on the corresponding transmission modes.
 This is a \emph{restless bandit system} because there can also be transitions for non-selected users (specifically, it is possible to transition from inactive to active).

To control the downloading process, there is a
central server with only $M$ threads ($M<N$), meaning that \emph{at most $M$ jobs can be processed simultaneously}. So at each time slot, the server has to make decisions selecting at most $M$ out of $N$ users to transmit a portion of their files. These decisions are further restricted by a global time average power constraint. The goal is to maximize the aggregate throughput, which is defined as
\begin{equation}
    \liminf_{T\rightarrow\infty}\frac{1}{T}\sum_{t=0}^{T-1}\sum_{n=1}^Nc_{n}\overline{B}_n\phi(\alpha_{n}(t)) \nonumber
\end{equation}
where $c_1, c_2, \ldots, c_N$ are a collection of positive weights that can be used to prioritize users.
Thus, this multi-user file downloading problem reduces  to the following:
\begin{align}
   \mbox{Max:} &~\liminf_{T\rightarrow\infty}\frac{1}{T}\sum_{t=0}^{T-1}\sum_{n=1}^N
    c_{n}\overline{B}_n\phi_n(\alpha_{n}(t)) \label{eq:multi-1} \\
    \mbox{S.t.:}&~\limsup_{T\rightarrow\infty}\frac{1}{T}\sum_{t=0}^{T-1}\sum_{n=1}^Np_n(\alpha_{n}(t))\leq \beta \label{eq:multi-2}\\
    &~\sum_{n=1}^NI(\alpha_{n}(t))\leq M~~\forall t\in\{0,1,2,\cdots\} \label{eq:mutli-3} \\
    &~Pr[F_n(t+1)=1~|~F_n(t)=0]=\lambda_n\label{eq:multi-4}\\
    &~Pr[F_n(t+1)=0~|~F_n(t)=1]=\phi_n(\alpha_n(t))\label{eq:multi-5}
\end{align}
where the constraints \eqref{eq:multi-4}-\eqref{eq:multi-5} hold for all $n \in \{1, \ldots, N\}$ and $t \in \{0, 1, 2,\ldots\}$, and
where $I(\cdot)$ is the indicator function defined as:
\begin{equation}
    I(x)=\left\{
           \begin{array}{ll}
             0, & \hbox{if $x=0$;} \\
             1, & \hbox{otherwise.}
           \end{array}
         \right.\nonumber
\end{equation}

\subsection{Lyapunov indexing algorithm} \label{multi:indexing}
This section develops our indexing algorithm for the multi-user case using the single-user case as a stepping stone. The major difficulty is  the instantaneous constraint $\sum_{n=1}^NI(\alpha_{n}(t))\leq M$.  Temporarily neglecting this constraint,
we use Lyapunov optimization to deal with the time average power constraint first.

We introduce a virtual queue $Q(t)$, which is again 0 at $t=0$. Instead of updating it on a frame basis, the server updates this queue every slot as follows:
\begin{equation}\label{m2}
    Q(t+1)=\max\left\{Q(t)+\sum_{n=1}^Np_n(\alpha_n(t))-\beta,0\right\}.
\end{equation}
Define $\script{N}(t)$ as the set of users beginning their renewal frames at time $t$, so that $F_n(t)=1$ for all such users.
 In general, $\script{N}(t)$ is a subset of $\script{N}=\{1,2,\cdots,N\}$. Define $|\script{N}(t)|$ as the number of users
 in the set $\script{N}(t)$.

At each time slot $t$, the server observes the queue state $Q(t)$ and chooses $(\alpha_1(t), \ldots, \alpha_N(t))$
in a manner similar to the single-user case. Specifically, for each user $n \in \script{N}(t)$ define:
\begin{equation}
    g_n(\alpha_n(t))\triangleq\frac{Vc_n\overline{B}_n\phi_n(\alpha_{n}(t))- Q(t)p_n(\alpha_{n}(t))}
        {1+\phi_n(\alpha_n(t))/\lambda_n} \label{eq:multi-alg1}
\end{equation}
This
is similar to the expression \eqref{e5} used in the single-user optimization.   Call $g_n(\alpha_n(t))$ a \emph{reward}.
Now define an index for each subsystem $n$ by:
\begin{equation}\label{m5}
    \gamma_n(t)\triangleq\max_{\alpha_n(t)\in\mathcal{A}_n}g_n(\alpha_n(t))
\end{equation}
which is the maximum possible reward one can get from the $n$-th subsystem at time slot $t$. Thus, it is natural to define the following myopic algorithm:  Find the (at most) $M$ subsystems in $\script{N}(t)$ with the greatest rewards, and serve these with their corresponding
optimal $\alpha_n(t)$ options in $\script{A}_n$ that maximize $g_n(\alpha_n(t))$.
\begin{algorithm}
\begin{Alg}~
\begin{itemize}
  \item At each time slot $t$, the server observes virtual queue state $Q(t)$ and computes the indices using (\ref{m5}) for all $n\in\mathcal{N}(t)$.
  \item Activate the $\min[M, |\script{N}(t)|]$ subsystems with greatest indices, using their corresponding actions $\alpha_n(t) \in \script{A}_n$ that
  maximize $g_n(\alpha_n(t))$.
  \item Update $Q(t)$ according to \eqref{m2}  at the end of each slot $t$.
\end{itemize}
\end{Alg}
\end{algorithm}

\subsection{Theoretical performance analysis}

In this subsection, we show that the above algorithm always satisfies the desired time average power constraint.  Define:
\begin{eqnarray*}
p^{min}_n &=& \min_{\alpha_n\in\mathcal{A}_n\setminus\{0\}}p_n(\alpha_n) \\
p^{min} &=& \min_np^{min}_n \\
p^{max}_n &=& \max_{\alpha_n\in\mathcal{A}_n}p_n(\alpha_n) \\
c^{max} &=& \max_{n}c_n \\
\overline{B}^{max} &=& \max_n \overline{B}_n
\end{eqnarray*}
Assume that $p^{min}>0$.

\begin{lemma} \label{lem:3}
Under the above Lyapunov indexing algorithm, the
queue $\{Q(t)\}_{t=0}^{\infty}$  is deterministically bounded. Specifically, we have for all $t \in \{0, 1, 2, \ldots\}$:
\[ Q(t)\leq \max\left\{\frac{Vc^{max}\overline{B}^{max}}{p^{min}}+\sum_{n=1}^Np^{max}_n-\beta,0\right\}\]
\end{lemma}
\begin{proof}
First, consider the case when $\sum_{n=1}^Np^{max}_n \leq \beta$. Since $Q(0)=0$, it is clear from the updating
rule \eqref{m2} that $Q(t)$ will remain 0 for all $t$.

Next, consider the case when $\sum_{n=1}^Np^{max}_n > \beta$. We prove the assertion by induction on $t$. The result
trivially holds for $t=0$. Suppose at $t=t'$, we have:
\[ Q(t')\leq\frac{Vc^{max}\overline{B}^{max}}{p^{min}}+\sum_{n=1}^Np^{max}_n-\beta \]
We are going to prove that the same statement holds for $t=t'+1$. We further divide it into two cases:
\begin{enumerate}
  \item $Q(t')\leq\frac{Vc^{max}\overline{B}^{max}}{p^{min}}$. In this case, since the queue increases by at most
  $\sum_{n=1}^Np^{max}_n - \beta$ on one slot, we have:
  \[ Q(t'+1) \leq \frac{Vc^{max}\overline{B}^{max}}{p^{min}} + \sum_{n=1}^Np^{max}_n - \beta\]

  \item $\frac{Vc^{max}\overline{B}^{max}}{p^{min}}<Q(t')\leq\frac{Vc^{max}\overline{B}^{max}}{p^{min}}+\sum_{n=1}^Np^{max}_n-\beta$. In this case, since $\phi_n(\alpha_n(t'))\leq 1$, there is no possibility that $Vc_n\overline{B}_n\phi_n(\alpha_n(t'))\geq Q(t')p_n(\alpha_n(t'))$ unless $\alpha_n(t')=0$. Thus, the Lyapunov indexing algorithm of minimizing
  \eqref{eq:multi-alg1} chooses  $\alpha_n(t')=0$ for all $n$. Thus, all indices are 0. This implies that $Q(t'+1)$ cannot increase, and we get $Q(t'+1)\leq\frac{Vc^{max}\overline{B}^{max}}{p^{min}}+\sum_{n=1}^Np^{max}_n-\beta$.
\end{enumerate}
\end{proof}

\begin{theorem}
The proposed Lyapunov indexing algorithm achieves the constraint:
\[ \limsup_{T\rightarrow\infty}\frac{1}{T}\sum_{t=0}^{T-1}\sum_{n=1}^Np_n(\alpha_{n}(t))\leq \beta \]
\end{theorem}
\begin{proof}
Using Lemma \ref{lem:1} under the special case that each frame only occupies one slot, we get that if $\{Q(t)\}_{t=0}^\infty$ is deterministically bounded, then the time average constraint is satisfied. Then, according to Lemma \ref{lem:3} we are done.
\end{proof}


\section{Multi-user optimality in a special case}\label{section:coupling}
 In general, it is very difficult to prove optimality of the above multi-user algorithm.
There are mainly two reasons.
The first reason is that multiple users
might renew themselves asynchronously, making it
difficult to define a ``renewal frame'' for the whole system.
Thus, the proof technique in Theorem 1 is infeasible.
The second reason is that, even without the time average constraint, the problem
degenerates into a standard restless bandit problem where the optimality of
 indexing is not guaranteed.

This section considers
a special case of the multi-user file downloading problem where
the Lyapunov indexing algorithm is provably optimal.
The special case has no time average power constraint.  Further,
for each user $n \in \{1, \ldots, N\}$:
\begin{itemize}
  \item Each file consists of a random number of fixed length packets with mean $\overline{B}_n=1/\mu_n$.
  \item The decision set $\mathcal{A}_n= \{0,1\}$, where 0 stands for ``idle'' and 1 stands for ``download.'' If $\alpha_n(t)=1$, then user $n$ successfully downloads a single packet.
  \item $\phi_n(\alpha_n(t)) = \mu_n\alpha_n(t)$.
  \item Idle time is geometrically distributed with mean $1/\lambda_n$.
  \item The special case \emph{$\mu_n = 1-\lambda_n$ is assumed}.
\end{itemize}
{\bf The assumption that the file length and idle time parameters $\mu_n$ and $\lambda_n$
satisfy $\mu_n = 1-\lambda_n$ is restrictive. However, when this assumption holds,
there exists a certain queueing system which admits \emph{exactly the same Markov dynamics} as the system considered here
(described in Section \ref{subsection:single_buffer} below). More importantly, it allows us to implement the stochastic coupling idea to prove  optimality.}

The goal is to maximize the sum throughput (in units of packets/slot), which is defined as:
\begin{equation}\label{thput}
    \liminf_{T\rightarrow\infty}\frac{1}{T}\sum_{t=0}^{T-1}\sum_{n=1}^N \overline{B}_n\phi(\alpha_{n}(t)).
\end{equation}
In this special case, the multi-user file downloading problem reduces to the following:
\begin{align}
   \mbox{Max:} &~\liminf_{T\rightarrow\infty}\frac{1}{T}\sum_{t=0}^{T-1}\sum_{n=1}^N\alpha_{n}(t) \label{sp:multi-1} \\
    \mbox{S.t.:}&~\sum_{n=1}^N\alpha_{n}(t)\leq M~~\forall t\in\{0,1,2,\cdots\} \label{sp:mutli-2} \\
        &~\alpha_n(t) \in \{0, F_n(t)\} \label{sp:multi-2b}\\
    &~Pr[F_n(t+1)=1~|~F_n(t)=0]=\lambda_n\label{sp:multi-3}\\
    &~Pr[F_n(t+1)=0~|~F_n(t)=1]=\alpha_{n}(t)(1-\lambda_n)\label{sp:multi-4}
\end{align}
where the equality \eqref{sp:multi-4} uses the fact that $\mu_n = 1-\lambda_n$.
A picture that illustrates the Markov structure
of constraints \eqref{sp:multi-2b}-\eqref{sp:multi-4} is given in Fig. \ref{fig:WTF1}

\subsection{A system with $N$ single-buffer queues}\label{subsection:single_buffer}

The above model, with the assumption $\mu_n = 1-\lambda_n$, is structurally equivalent to the following:
Consider a system of $N$ single-buffer queues, $M$ servers, and independent Bernoulli packet arrivals with rates $\lambda_n$
to each queue $n \in \{1, \ldots, N\}$.  This considers \emph{packet arrivals} rather than \emph{file arrivals}, so there are no file length variables and no parameters $\mu_n$ in this interpretation.  Let $\mathbf{A}(t) = (A_1(t), \ldots, A_N(t))$ be the binary-valued vector of packet arrivals on slot $t$, assumed to be i.i.d. over slots and independent in each coordinate.
 Assume all packets have the same size and each queue has a single buffer that can store just one packet.
 Let $F_n(t)$ be 1 if queue $n$ has a packet at the beginning of slot $t$, and $0$ else.  Each server can transmit at most 1 packet per slot.  Let $\alpha_n(t)$ be 1 if queue $n$ is served on slot $t$, and $0$ else.
 An arrival $A_n(t)$ occurs at the end of slot $t$ and is accepted only if queue $n$ is empty at the end of the slot (such as when it was served on that slot).  Packets that are not accepted are dropped.
 The Markov dynamics are described by the same figure as before, namely, Fig. \ref{fig:WTF1}.  Further,
the problem of maximizing throughput is given by the \emph{same} equations \eqref{sp:multi-1}-\eqref{sp:multi-4}. Thus, although the variables of the two problems have different interpretations, the problems are structurally equivalent.
For simplicity of exposition, the remainder of this section uses
this single-buffer queue interpretation.

\begin{figure}[top]
   \centering
   \includegraphics[height=2.5in]{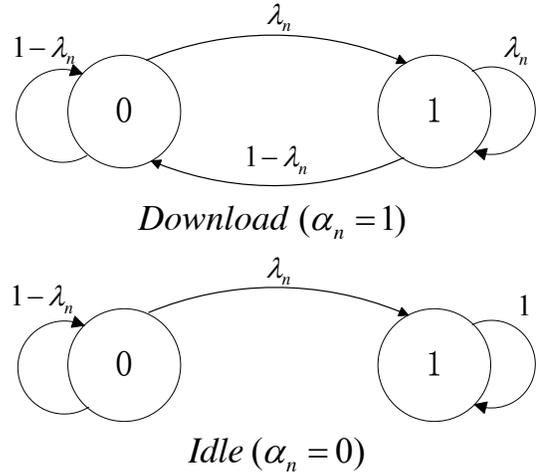} 
   \caption{Markovian dynamics of the $n$-th system.}
   \label{fig:WTF1}
\end{figure}

\subsection{Optimality of the indexing algorithm}
Since there is no power constraint, for any $V>0$ the
Lyapunov indexing policy \eqref{m5} in Section \ref{multi:indexing} reduces
to the following (using $c_n=1$, $Q(t)\equiv0$):   If there are fewer than $M$ non-empty queues, serve all of them. Else, serve the $M$ non-empty queues with the largest values of $\gamma_n$, where:
\begin{equation}
    \gamma_n=\frac{1}{1+(1-\lambda_n)/\lambda_n}=\lambda_n. \nonumber
\end{equation}
Thus, the Lyapunov indexing algorithm in this context reduces to serving the (at most $M$) non-empty queues with the  largest $\lambda_n$ values each time slot.  For the remainder of this section, this is called the Max-$\lambda$ policy.
The following theorem shows that Max-$\lambda$ is optimal in this context.

\begin{theorem} \label{thm:max-lambda}
The Max-$\lambda$ policy is optimal for the problem  \eqref{sp:multi-1}-\eqref{sp:multi-4}.  In particular, under the single-buffer queue interpretation, it maximizes throughput over all policies that transmit on each slot $t$
without knowledge of the arrival vector
$\mathbf{A}(t)$.
\end{theorem}

For the $N$ single-buffer queue interpretation,
the total throughput is equal to the raw arrival rate $\sum_{i=1}^N \lambda_i$ minus the packet drop rate.
Intuitively, the reason Max-$\lambda$ is optimal is that it chooses to leave packets in the queues that are least likely to induce packet drops. An example comparison of the throughput gap between Max-$\lambda$ and Min-$\lambda$ policies is given in Appendix A.

The proof of Theorem \ref{thm:max-lambda} is divided into two parts. The first part uses stochastic coupling techniques to prove that Max-$\lambda$ dominates all alternative \emph{work-conserving} policies.  A policy is \emph{work-conserving} if it does not allow any server to be idle when it could be used to serve a non-empty queue.
 The second part of the proof
 shows that throughput cannot be increased by considering non-work-conserving policies.

\subsection{Preliminaries on stochastic coupling}\label{prelim}

Consider two discrete time processes $\mathcal{X}\triangleq\{X(t)\}_{t=0}^{\infty}$ and $\mathcal{Y}\triangleq\{Y(t)\}_{t=0}^{\infty}$.  The notation $\mathcal{X} =_{st} \mathcal{Y}$ means that $\mathcal{X}$ and $\mathcal{Y}$ are \emph{stochastically equivalent}, in that they are described by the same probability law.  Formally, this means that
their joint distributions are the same, so for all $t \in \{0, 1, 2, \ldots\}$ and all $(z_0, \ldots, z_t) \in \mathcal{R}^{t+1}$:
\begin{eqnarray*}
&&Pr[X(0)\leq z_0, \ldots, X(t)\leq z_t] \\
&&= Pr[Y(0)\leq z_0, \ldots, Y(t) \leq z_t]
\end{eqnarray*}
The notation $\mathcal{X}\leq_{st}\mathcal{Y}$ means that $\mathcal{X}$ is \emph{stochastically less than or equal to} $\mathcal{Y}$, as defined by the following theorem.
\begin{theorem}  \label{sto_theorem}
(\cite{tass-server-allocation}) The following three statements are equivalent:
\begin{enumerate}
  \item $\mathcal{X}\leq_{st}\mathcal{Y}$.
  \item $Pr[g(X(0), X(1), \cdots, X(t))>z]\leq Pr[g(Y(0),$ $Y(1), \cdots, Y(t))>z]$ for all $t\in \mathbb{Z}^{+}$, all $z$, and for all functions $g:\mathcal{R}^{n}\rightarrow\mathcal{R}$ that are measurable and nondecreasing in all coordinates.
  \item There exist two stochastic processes $\mathcal{X}'$ and $\mathcal{Y}'$ on a common  probability space that satisfy $\mathcal{X}=_{st}\mathcal{X}'$, $\mathcal{Y}=_{st}\mathcal{Y}'$, and $X'(t)\leq Y'(t)$ for every $t\in \mathbb{Z}^{+}$.
\end{enumerate}
\end{theorem}

The following additional notation is used in the proof of Theorem \ref{thm:max-lambda}.

\begin{itemize}
  \item Arrival vector $\{\mathbf{A}(t)\}_{t=0}^{\infty}$, where $\mathbf{A}(t)\triangleq[A_1(t)$ $A_2(t)~\cdots~A_N(t)]$. Each $A_n(t)$ is an independent binary random variable that takes $1$ w.p. $\lambda_n$ and $0$ w.p. $1-\lambda_n$.
  \item Buffer state vector $\{\mathbf{F}(t)\}_{t=0}^{\infty}$, where $\mathbf{F}(t)\triangleq[F_1(t)$ $F_2(t)~\cdots~F_N(t)]$. So $F_n(t)=1$ if queue $n$ has a packet at the beginning of slot $t$, and $F_n(t)=0$ else.
  \item Total packet process $\mathcal{U}\triangleq\{U(t)\}_{t=0}^{\infty}$, where $U(t)\triangleq\sum_{n=1}^NF_n(t)$  represents the total number of packets in the system on slot $t$.  Since each queue can hold at most one packet, we have $0 \leq U(t) \leq N$ for all slots $t$.
\end{itemize}

\subsection{Stochastic ordering of buffer state process}

The next lemma is the key to proving Theorem \ref{thm:max-lambda}.  The lemma considers the multi-queue system with a fixed but arbitrary initial buffer state $\mathbf{F}(0)$. The arrival process $\mathbf{A}(t)$ is as defined above.
Let $\mathcal{U}^{\mbox{\tiny Max-$\lambda$}}$ be the total packet process under the Max-$\lambda$ policy.
Let $\mathcal{U}^{\pi}$ be the corresponding process starting from the same initial state $\mathbf{F}(0)$ and having the same arrivals $\mathbf{A}(t)$, but with an arbitrary work-conserving policy $\pi$.

\begin{lemma} \label{sto_file_stat}
The total packet processes $\mathcal{U}^\pi$ and $\mathcal{U}^{\mbox{\tiny Max-$\lambda$}}$ satisfy:
\begin{equation}\label{st_order}
   \mathcal{U}^{\pi}\leq_{st}\mathcal{U}^{\mbox{\tiny Max-$\lambda$}}
\end{equation}
\end{lemma}
\begin{proof}
Without loss of generality, assume the queues are sorted so that $\lambda_n\leq\lambda_{n+1},~n=1,2,\cdots,N-1$. Define $\{\mathbf{F}^\pi(t)\}_{t=0}^{\infty}$ as the buffer state vector under policy $\pi$.  Define
$\{\mathbf{F}^{\mbox{\tiny Max-$\lambda$}}(t)\}_{t=0}^{\infty}$
as the corresponding buffer states under the Max-$\lambda$ policy.
By assumption the initial states satisfy $\mathbf{F}^\pi(0)=\mathbf{F}^{\mbox{\tiny Max-$\lambda$}}(0)$.  Next,  we construct a \emph{third process} $\mathcal{U}^\lambda$
with a \emph{modified} arrival vector process $\{\mathbf{A}^\lambda(t)\}_{t=0}^{\infty}$ and a corresponding
buffer state vector $\{\mathbf{F}^\lambda(t)\}_{t=0}^{\infty}$ (with the same initial state
$\mathbf{F}^\lambda(0)=\mathbf{F}^\pi(0)$), which satisfies:
\begin{enumerate}
  \item $\mathcal{U}^\lambda$ is also generated from the Max-$\lambda$ policy.
  \item $\mathcal{U}^\lambda =_{st} \mathcal{U}^{\mbox{\tiny Max-$\lambda$}}$. Since the total packet process is completely determined by the initial state, the scheduling policy, and the arrival process, it suffices to show that $\{\mathbf{A}^\lambda(t)\}_{t=0}^{\infty}$ and $\{\mathbf{A}(t)\}_{t=0}^{\infty}$ have the same probability law.
  \item $U^\pi(t)\leq U^\lambda(t)~\forall t\geq0$.
\end{enumerate}

Since the arrival process $\mathbf{A}(t)$ is i.i.d. over slots, in order to guarantee 2) and 3), it is sufficient to construct $\mathbf{A}^\lambda(t)$ coupled with $\mathbf{A}(t)$ for each $t$ so that the following two properties hold for all $t\geq0$:
\begin{itemize}
  \item  The random variables $\mathbf{A}(t)$ and $\mathbf{A}^\lambda(t)$ have the same probability law. Specifically, both  produce arrivals according to Bernoulli processes that are independent over queues and over time, with $Pr[A_n(t)=1]=Pr[A_n^\lambda(t)=1] = \lambda_n$ for all $n \in \{1, \ldots, N\}$.
  \item For all $j\in\{1,2,\cdots,N\}$,
  \begin{equation}\label{partial_order}
            \sum_{n=1}^jF^\pi_{n}(t)\leq\sum_{n=1}^jF^\lambda_{n}(t),
  \end{equation}
\end{itemize}

The construction is based on an induction.

At $t=0$ we have $\mathbf{F}^\pi(0)=\mathbf{F}^\lambda(0)$. Thus, \eqref{partial_order} naturally holds for $t=0$.
Now fix $\tau \geq 0$ and assume \eqref{partial_order} holds for all slots up to time $t=\tau$.  If $\tau\geq 1$, further assume the arrivals $\{\mathbf{A}^\lambda(t)\}_{t=0}^{\tau-1}$ have been constructed to have the same probability law as $\{\mathbf{A}(t)\}_{\tau=0}^{\tau-1}$. Since arrivals on slot $\tau$ occur at the \emph{end} of slot $\tau$, the arrivals $\mathbf{A}^\lambda(\tau)$ must be constructed.
We are going to show  there exists an $\mathbf{A}^\lambda(\tau)$ that is \emph{coupled}  with $\mathbf{A}(\tau)$ so that it has the same probability law and it also ensures \eqref{partial_order} holds for $t=\tau+1$.

Since arrivals occur after the transmitting action, we divide the analysis into two parts. First, we analyze the temporary buffer states after the transmitting action \emph{but before arrivals occur}. Then, we define arrivals $\mathbf{A}^\lambda(\tau)$ at the end of slot $\tau$ to achieve the desired coupling.

Define $\tilde{\mathbf{F}}^\pi(\tau)$ and $\tilde{\mathbf{F}}^{\lambda}(\tau)$ as the \emph{temporary buffer states} right after the transmitting action at slot $\tau$ but before arrivals occur under policy $\pi$ and policy Max-$\lambda$, respectively.  Thus, for each queue $n \in \{1, \ldots, N\}$:
\begin{eqnarray}
\tilde{F}_n^\pi(\tau) &=& F_n^\pi(\tau) - \alpha_n^\pi(\tau) \label{eq:f-pi-dude} \\
\tilde{F}_n^\lambda(\tau) &=& F_n^\lambda(\tau) - \alpha_n^\lambda(\tau) \label{eq:f-lambda-dude}
\end{eqnarray}
where $\alpha_n^\pi(\tau)$ and $\alpha_n^\lambda(\tau)$ are the slot $\tau$ decisions under policy
$\pi$ and Max-$\lambda$, respectively. Since
 \eqref{partial_order} holds for $j=N$ on slot $\tau$, the total number of packets at the start of slot $\tau$ under policy $\pi$ is less than or equal to that under Max-$\lambda$.   Since both policies $\pi$ and Max-$\lambda$ are work-conserving,
 it is impossible for policy $\pi$ to transmit more packets than Max-$\lambda$ during slot $\tau$.
 This implies:
  \begin{equation}\label{FpileqFlambda}
\sum_{n=1}^N\tilde{F}_n^\pi(\tau)\leq\sum_{n=1}^N\tilde{F}^\lambda_{n}(\tau).
\end{equation}
 Indeed, if $\pi$ transmits the same number of packets as Max-$\lambda$ on slot $\tau$, then \eqref{FpileqFlambda} clearly holds.  On the other hand, if $\pi$ transmits \emph{fewer} packets than Max-$\lambda$, it must transmit fewer than $M$ packets (since $M$ is the number of servers).  In this case, the work-conserving nature of $\pi$ implies that all non-empty queues were served, so that $\tilde{F}_n^\pi(\tau)=0$ for all $n$ and \eqref{FpileqFlambda} again holds.
We now claim the following holds:
\begin{lemma}  \label{lem:new}
\begin{equation}\label{temp_order}
\sum_{n=1}^j\tilde{F}_n^\pi(\tau)\leq\sum_{n=1}^j\tilde{F}^\lambda_{n}(\tau)~~\forall j\in\{1,2,\cdots,N\}.
\end{equation}
\end{lemma}
\begin{proof}
See Appendix B.
\end{proof}

Now
let $j^\pi(l)$ and $j^\lambda(l)$ be the subscript of $l$-th \emph{empty} temporary buffer (with order starting from the first queue) corresponding to $\tilde{\mathbf{F}}^\pi(\tau)$ and $\tilde{\mathbf{F}}^{\lambda}(\tau)$, respectively.
It  follows from \eqref{temp_order} that the $\pi$ system on slot $\tau$ has at  least as many empty temporary buffer states as the Max-$\lambda$ policy, and:
\begin{equation}\label{empty_order}
    j^\pi(l)\leq j^\lambda(l)~~\forall l\in\{1,2,\cdots,K(\tau)\}
\end{equation}
where $K(\tau)\leq N$ is the the number of empty temporary buffer states under Max-$\lambda$ at time slot $\tau$.
Since $\lambda_i\leq\lambda_j$ if and only if $i\leq j$, \eqref{empty_order} further implies that
\begin{equation}\label{lambda_order}
    \lambda_{j^\pi(l)}\leq\lambda_{j^\lambda(l)}~~\forall l\in\{1,2,\cdots,K(\tau)\}.
\end{equation}

Now construct the arrival vector $\mathbf{A}^\lambda(\tau)$ for the system with the Max-$\lambda$ policy  in the following way:
\begin{align}
    A_{j^\pi(l)}(\tau)=1&\Rightarrow A^\lambda_{j^\lambda(l)}(\tau)=1 ~~~w.p.~1\label{coupling1}\\
    A_{j^\pi(l)}(\tau)=0&\Rightarrow\left\{
                                  \begin{array}{ll}
                                    A^\lambda_{j^\lambda(l)}(\tau)=0 , & \hbox{w.p.~$\frac{1-\lambda_{j^\lambda(l)}}{1-\lambda_{j^\pi(l)}}$;} \\
                                    A^\lambda_{j^\lambda(l)}(\tau)=1,  & \hbox{w.p.~$\frac{\lambda_{j^\lambda(l)}-\lambda_{j^\pi(l)}}{1-\lambda_{j^\pi(l)}}$.}
                                  \end{array}
                                \right.\label{coupling2}
\end{align}
Notice that \eqref{coupling2} uses valid probability distributions because of \eqref{lambda_order}.
This establishes the slot $\tau$ arrivals for the Max-$\lambda$ policy for all of its $K(\tau)$ queues with empty temporary buffer states.  The slot $\tau$ arrivals for its queues with non-empty temporary buffers
will be dropped and hence do not affect the queue states on slot $\tau+1$.  Thus, we define arrivals $A_j^\lambda(\tau)$ to be independent of all other quantities and to be Bernoulli with $Pr[A_j^\lambda(\tau)=1]=\lambda_j$ for all $j$ in the set:
\[ j\in\{1,2,\cdots,N\}\setminus\{j^\lambda(1),\cdots,j^\lambda(K(\tau))\} \]
Now we verify that $\mathbf{A}(\tau)$ and $\mathbf{A}^\lambda(\tau)$ have the same probability law.
First condition on knowledge of $K(\tau)$ and the particular $j^\pi(l)$ and $j^\lambda(l)$ values for $l \in \{1, \ldots, K(\tau)\}$.
All queues $j$ with non-empty temporary buffer states on slot $\tau$ under Max-$\lambda$ were defined to have arrivals $A_{j}^\lambda(\tau)$ as independent Bernoulli variables with $Pr[A_j^\lambda(\tau)=1]=\lambda_j$.
It remains to  verify those queues within $\{j^\lambda(1),\cdots,j^\lambda(K(\tau))\}$.
According to \eqref{coupling2}, for any queue $j^\lambda(l)$ in set  $\{j^\lambda(1),\cdots,j^\lambda(K(\tau))\}$, it follows
\begin{eqnarray*}
Pr\left[A^\lambda_{j^\lambda(l)}(\tau)=0\right]&=&(1-\lambda_{j^\pi(l)})\frac{1-\lambda_{j^\lambda(l)}}{1-\lambda_{j^\pi(l)}}\nonumber\\
&=&1-\lambda_{j^\lambda(l)}
\end{eqnarray*}
and so $Pr[A_j^\lambda(\tau)=1]=\lambda_j$ for all $j \in \{j^\lambda(l)\}_{l=1}^{K(\tau)}$.
Further, mutual independence of $\{A_{j^\pi(l)}(\tau)\}_{l=1}^{K(\tau)}$ implies mutual independence of
$\{A_{j^\lambda(l)}(\tau)\}_{l=1}^{K(\tau)}$.   Finally, these quantities are conditionally
independent of events before slot $\tau$, given knowledge of $K(\tau)$ and the particular $j^\pi(l)$ and $j^\lambda(l)$ values for $l \in \{1, \ldots, K(\tau)\}$.
Thus, conditioned on this knowledge,
 $\mathbf{A}(\tau)$ and $\mathbf{A}^\lambda(\tau)$ have the same probability law.  This holds for all possible values of the conditional knowledge $K(\tau)$  and $j^\pi(l)$ and $j^\lambda(l)$.  It follows that $\mathbf{A}(\tau)$ and $\mathbf{A}^\lambda(\tau)$ have the same (unconditioned) probability law.

Finally, we show that the coupling relations \eqref{coupling1} and \eqref{coupling2} produce such $\mathbf{F}^\lambda(\tau+1)$ satisfying
\begin{equation}\label{conclusion}
    \sum_{n=1}^jF^\pi_{n}(\tau+1)\leq\sum_{n=1}^jF^\lambda_{n}(\tau+1),~\forall~j\in\{1,2,\cdots,N\}.
\end{equation}
According to \eqref{coupling1} and \eqref{coupling2},
\[A_{j^\pi(l)}(\tau)\leq A^\lambda_{j^\lambda(l)}(\tau),~~\forall l\in\{1,\cdots,K(\tau)\},\]
thus,
\begin{equation}\label{arrival_ineq}
\sum_{i=1}^{l}A_{j^\pi(i)}(\tau)\leq \sum_{i=1}^{l}A^\lambda_{j^\lambda(i)}(\tau),~~\forall l\in\{1,\cdots,K(\tau)\}.
\end{equation}
Pick any $j\in\{1,2,\cdots,N\}$. Let $l^\pi$ be the number of empty temporary buffers within the first $j$ queues
under policy $\pi$, i.e.
\[l^\pi=\max_{j^\pi(l)\leq j}l\]
Similarly define:
\[l^\lambda=\max_{j^\lambda(l)\leq j}l.\]
Then, it follows:
\begin{eqnarray}
\sum_{n=1}^jF^\pi_{n}(\tau+1)&=&\sum_{n=1}^j\tilde{F}^\pi_{n}(\tau)+\sum_{i=1}^{l^\pi}A_{j^\pi(i)}(\tau) \label{F(t+1)}\\
\sum_{n=1}^jF^\lambda_{n}(\tau+1)&=&\sum_{n=1}^j\tilde{F}^\lambda_{n}(\tau)+\sum_{i=1}^{l^\lambda}A^\lambda_{j^\lambda(i)}(\tau)\label{barF(t+1)}
\end{eqnarray}
We know that $l^\pi \geq l^\lambda$.  So there are two cases:
\begin{itemize}
  \item If $l^\pi=l^\lambda$, then from \eqref{F(t+1)}:
  \begin{eqnarray*}
  \sum_{n=1}^jF_n^\pi(\tau+1) &=& \sum_{n=1}^j\tilde{F}_n^\pi(\tau) + \sum_{i=1}^{l^{\lambda}}A_{j^\pi(i)}(\tau)\\
  &\leq& \sum_{n=1}^j\tilde{F}_n^\lambda(\tau) + \sum_{i=1}^{l^{\lambda}}A_{j^{\lambda}(i)}(\tau) \\
  &=& \sum_{n=1}^j F_n^\lambda(\tau+1)
  \end{eqnarray*}
where the inequality follows from \eqref{temp_order} and from \eqref{arrival_ineq} with
$l=l^{\lambda}$.
Thus, \eqref{conclusion} holds.
  \item If $l^\pi>l^\lambda$, then from \eqref{F(t+1)}:
\begin{align}
\sum_{n=1}^jF^\pi_{n}(\tau+1)=&\sum_{n=1}^j\tilde{F}^\pi_{n}(\tau)+\sum_{i=1}^{l^\lambda}A_{j^\pi(i)}(\tau)\nonumber\\
&+\sum_{i=l^\lambda+1}^{l^\pi}A_{j^\pi(i)}(\tau)\nonumber\\
\leq&\sum_{n=1}^j\tilde{F}^\lambda_{n}(\tau)+\sum_{i=1}^{l^\lambda}A_{j^\pi(i)}(\tau)\nonumber\\
\leq&\sum_{n=1}^j\tilde{F}^\lambda_{n}(\tau)+\sum_{i=1}^{l^\lambda}A^\lambda_{j^\lambda(i)}(\tau) \nonumber\\
=&\sum_{n=1}^jF^\lambda_{n}(\tau+1).\nonumber
\end{align}
where the first inequality follows from the fact that
\begin{eqnarray*}
\sum_{i=l^\lambda+1}^{l^\pi}A_{j^\pi(i)}(\tau)&\leq& l^\pi-l^\lambda\\
&=&(j-l^\lambda)-(j-l^\pi)  \nonumber\\
&=&\sum_{n=1}^j\tilde{F}^\lambda_{n}(\tau)-\sum_{n=1}^j\tilde{F}^\pi_{n}(\tau),\nonumber
\end{eqnarray*}
and the second inequality follows from \eqref{arrival_ineq}.
\end{itemize}
Thus, \eqref{partial_order} holds for $t=\tau+1$ and the induction step is done.
\end{proof}

\begin{corollary} \label{corollary:1}
The Max-$\lambda$ policy maximizes throughput within the class of work-conserving policies.
\end{corollary}
\begin{proof}
Let $S^\pi(t)$ be the number of packets transmitted under any work-conserving policy $\pi$ on slot $t$,
and let $S^{\mbox{\tiny Max-$\lambda$}}(t)$ be the corresponding process under policy Max-$\lambda$.  Lemma
\ref{sto_file_stat} implies $\mathcal{U}^\pi(t) \leq_{st} \mathcal{U}^{\mbox{\tiny Max-$\lambda$}}$. Then:
\begin{eqnarray*}
 \expect{S^\pi(t)}&=&\expect{\min[U^\pi(t), M]}   \\
  &\leq&\expect{\min[U^{\mbox{\tiny Max-$\lambda$}}(t),M]} \\
  &=& \expect{S^{\mbox{\tiny Max-$\lambda$}}(t)}
\end{eqnarray*}
where the inequality follows from Theorem \ref{sto_theorem}, with the understanding that $g(U(0),\ldots, U(t))\triangleq\min[U(t),M]$ is a function that is nondecreasing in all coordinates.
\end{proof}

\subsection{Extending to non-work-conserving policies}

Corollary \ref{corollary:1} establishes optimality of Max-$\lambda$ over the class of all work-conserving policies.
To complete the proof of Theorem \ref{thm:max-lambda}, it remains to show that throughput cannot be increased by allowing for non-work-conserving policies.  It suffices to show that for any non-work-conserving policy, there exists a work-conserving policy that gets the same or better throughput.  The proof is straightforward and we give only a
proof sketch for brevity.
Consider any non-work-conserving policy $\pi$, and let $F_n^\pi(t)$ be its buffer state process on slot $t$ for each queue $n$.  For the same initial buffer state and arrival process, define the work-conserving policy $\pi'$ as follows:  Every slot $t$, policy $\pi'$ initially allocates the $M$ servers to exactly the same queues as policy $\pi$.  However, if some of these queues are empty under policy $\pi'$, it reallocates those servers to any non-empty queues that are not yet allocated servers (in keeping with the work-conserving property).  Let $F_n^{\pi'}(t)$ be the buffer state process for queue $n$ under policy $\pi'$.
 It is not difficult to show that $F_n^\pi(t) \geq F_n^{\pi'}(t)$ for all queues $n$ and all slots $t$.
 Therefore, on every slot $t$, the amount of \emph{blocked arrivals} under policy $\pi$ is always greater than or equal to that under policy $\pi'$.  This implies the throughput under policy $\pi$ is less than or equal to that of policy $\pi'$.

\section{Simulation experiments} \label{section:sims}

In this section, we demonstrate near optimality of the multi-user Lyapunov indexing algorithm by extensive simulations. In the first part, we simulate the case in which the file length distribution is geometric, and show that the suboptimality gap is extremely small. In the second part, we test the robustness of our algorithm for more general scenarios in which the file length distribution is not geometric.
For simplicity, it is assumed throughout that all transmissions send a fixed sized packet, all files are an integer number of these packets, and that decisions $\alpha_n(t) \in \script{A}_n$ affect the success probability of the transmission as well as the power expenditure.

\subsection{Lyapunov indexing with geometric file length} \label{subsection:geometric-sim}
In the first simulation we use $N=8$, $M=4$ with action set $\mathcal{A}_n=\{0,1\}~\forall n$; The settings are generated randomly and specified in Table I, and the constraint $\beta=5$.

\begin{table}
\begin{center}
\caption{Problem parameters}
\begin{tabular}{|l|l|l|l|l|l|l|}
  \hline
   User & $\lambda_n$ & $\mu_n$ & $\phi_n(1)$ & $c_n$ & $p_n(1)$  \\ \hline
  1 & 0.0028 & 0.5380 & 0.4842 & 4.7527 & 3.9504  \\
  2 & 0.4176 & 0.5453 & 0.4908 & 2.0681 & 3.7391  \\
  3 & 0.0888 & 0.5044 & 0.4540 & 2.8656 & 3.5753  \\
  4 & 0.3181 & 0.6103 & 0.5493 & 2.4605 & 2.1828  \\
  5 & 0.4151 & 0.9839 & 0.8855 & 4.5554 & 3.1982  \\
  6 & 0.2546 & 0.5975 & 0.5377 & 3.9647 & 3.5290  \\
  7 & 0.1705 & 0.5517 & 0.4966 & 1.5159 & 2.5226  \\
  8 & 0.2109 & 0.7597 & 0.6837 & 3.6364 & 2.5376  \\
  \hline
\end{tabular}
\end{center}
\end{table}

The algorithm is run for 1 million slots in each trial and each point is the average of 100 trials. We compare the performance of our algorithm with the optimal randomized policy. The optimal policy is computed by constructing composite states (i.e. if there are three users where user 1 is at state 0, user 2 is at state 1 and user 3 is at state 1, we view 011 as a composite state), and then reformulating this MDP into a linear program (see \cite{MDP}) with $\mathbf{5985}$ variables and $\mathbf{258}$ constraints.

In Fig. \ref{fig:Stupendous2}, we show that as our tradeoff parameter $V$ gets larger, the objective value approaches the optimal value and achieves a near optimal performance. Fig. \ref{fig:Stupendous3} and Fig. \ref{fig:Stupendous4} show that $V$ also affects the virtual queue size and the constraint gap. As $V$ gets larger, the average virtual queue size becomes larger and the gap becomes smaller. We also plot the upper bound of queue size we derived from Lemma \ref{lem:3} in Fig. \ref{fig:Stupendous3}, demonstrating that the queue is bounded. In order to show that $V$ is indeed a trade-off parameter affecting the convergence time, we plotted Fig. \ref{fig:Stupendous5}. It can be seen from the figure that as $V$ gets larger, the number of time slots needed for the running average to roughly converge to the optimal power expenditure becomes larger.

\begin{figure}[htbp]
   \centering
   \includegraphics[height=2.5in]{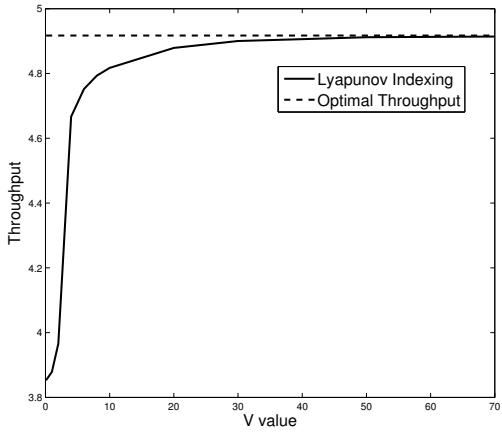} 
   \caption{Throughput versus tradeoff parameter V}
   \label{fig:Stupendous2}
\end{figure}

\begin{figure}[htbp]
   \centering
   \includegraphics[height=2.5in]{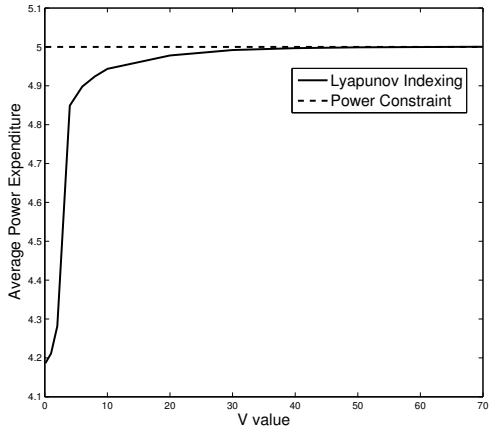} 
   \caption{The time average power consumption versus tradeoff parameter $V$.}
   \label{fig:Stupendous3}
\end{figure}

\begin{figure}[htbp]
   \centering
   \includegraphics[height=2.5in]{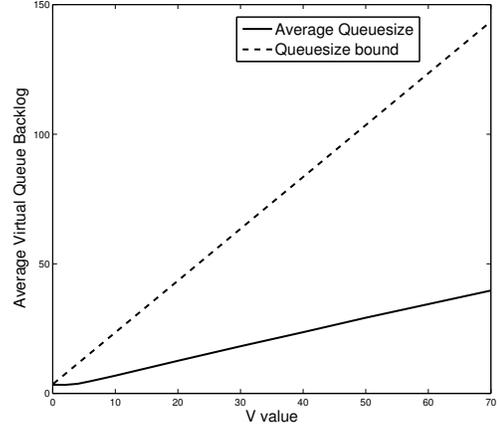} 
   \caption{Average virtual queue backlog versus tradeoff parameter $V$.}
   \label{fig:Stupendous4}
\end{figure}

\begin{figure}[htbp]
   \centering
   \includegraphics[height=2.5in]{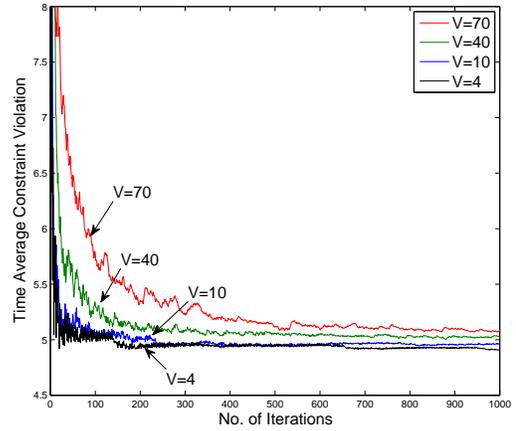} 
   \caption{Running average power consumption versus tradeoff parameter $V$.}
   \label{fig:Stupendous5}
\end{figure}

In the second simulation, we explore the parameter space and demonstrate that in general the suboptimality gap of our algorithm is negligible. First, we define the relative error as the following:
\begin{equation}\label{e20}
    \textrm{relative error}=\frac{|OBJ-OPT|}{OPT}
\end{equation}
where $OBJ$ is the objective value after running 1 million slots of our algorithm and $OPT$ is the optimal value. We first explore the system parameters by letting $\lambda_{n}$'s and $\mu_n$'s take random numbers within 0 and 1, letting $c_n$ take random numbers within 1 and 5, choosing $V=70$ and fixing the remaining
parameters the same as the last experiment. We conduct 1000 Monte-Carlo experiments and calculate the average relative error, which is \textbf{0.00083}.

Next, we explore the control parameters by letting  the $p_{n}(1)$ values take  random numbers within 2 and 4, letting $\phi_n(1)/\mu_n$ take  random numbers between 0 and 1, choosing $V=70$, and fixing the remaining parameters the same as the first simulation. The relative error is \textbf{0.00057}. Both experiments show that the suboptimality gap is extremely small.

\subsection{Lyapunov indexing with non-memoryless file lengths}
In this part, we test the sensitivity of the algorithm to different file length distributions. In particular, the uniform distribution and the Poisson distribution are implemented respectively, while our algorithm still treats them as a geometric distribution with the same mean. We then compare their throughputs with the geometric case.

We use $N=9$, $M=4$ with action set $\mathcal{A}_n=\{0,1\}~\forall n$. The settings are specified in Table II with constraint $\beta=5$. Notice that for geometric and uniform distribution, the file lengths are taken to be integer values. The algorithm is run for 1 million slots in each trial and each point is the average of 100 trials.
\begin{table}
\begin{center}
\caption{Problem parameters under geometric, uniform and poisson distribution}
\begin{tabular}{|l|l|l|l|l|l|l|l|}
  \hline
  User & $\mu_n$ & Unif. & Poiss. & $\lambda_n$ & $\phi_n(1)$ & $c_n$ & $p_n(1)$ \\
   &~ & interval & mean &~&~&~& \\ \hline
   1 & 1/3 & [1,5] & 3 & 0.4955 & 0.1832 & 4.3261 & 2.8763 \\
   2 & 1/2 & [1,3] & 2 & 0.1181 & 0.4187 & 1.6827 & 2.0549 \\
  3 & 1/2 & [1,3] & 2 & 0.1298 & 0.4491 & 1.9483 & 2.1469 \\
  4 & 1/7 & [1,13] & 7 & 0.4660 & 0.0984 & 2.7495 & 3.4472 \\
  5  & 1/4 & [1,7] & 4 & 0.1661 & 0.1742 & 1.5535 & 3.2801 \\
  6 & 1/3 & [1,5] & 3 & 0.2124 & 0.3101 & 4.3151 & 3.5648 \\
  7 & 1/2 & [1,3] & 2 & 0.5295 & 0.4980 & 3.6701 & 2.4680 \\
  8 & 1/5 & [1,9] & 5 & 0.2228 & 0.1971 & 4.0185 & 2.2984 \\
  9 & 1/4 & [1,7] & 4 & 0.0332 & 0.1986 & 3.0411 & 2.5747 \\
  \hline
\end{tabular}
\end{center}
\end{table}

While the decisions are made using these values, the affect of these decisions incorporates the actual (non-memoryless) file sizes.
Fig. \ref{fig:Stupendous6}  shows the throughput-versus-$V$ relation for the two non-memoryless cases and the memoryless case with matched means. The performance of all three is similar.  This illustrates that
the indexing algorithm is robust under different file length distributions.

\begin{figure}[htbp]
   \centering
   \includegraphics[height=2.5in]{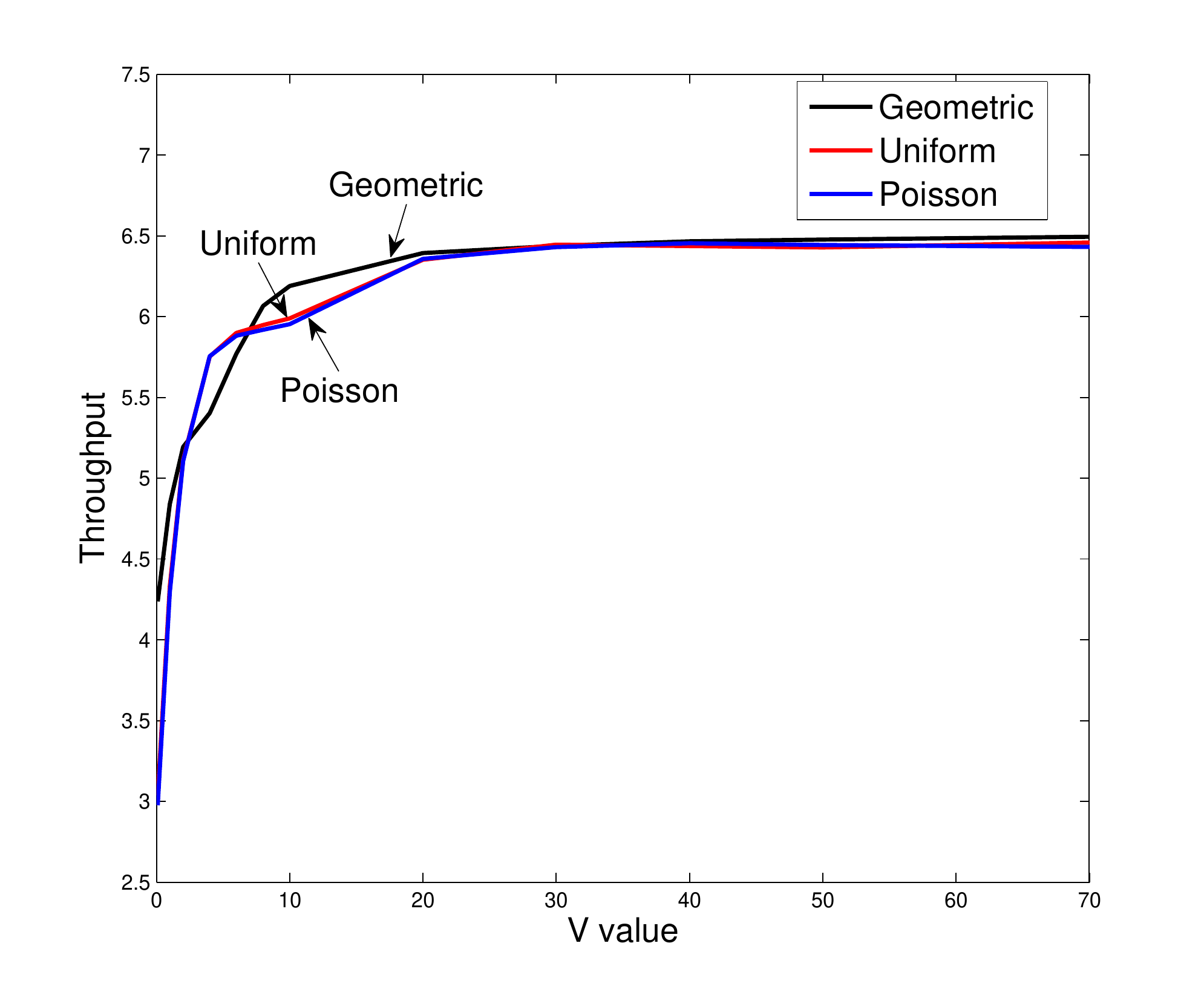} 
   \caption{Throughput versus tradeoff parameter $V$ under different file length distributions.}
   \label{fig:Stupendous6}
\end{figure}

\section{Conclusions}

We have investigated a file downloading system where the network delays affect the file arrival processes. The single-user case was solved by a variable frame length Lyapunov optimization  method.  The technique was extended as a well-reasoned heuristic algorithm for the multi-user case.  Such heuristics are important because the problem is a multi-dimensional Markov decision problem with very high complexity.
The heuristic is simple, can be implemented in an online fashion, and was analytically shown to achieve the desired average power constraint. Moreover, under a special case with no average power constraint,  stochastic coupling was used to prove the heuristic is throughput optimal. Simulations suggest that the algorithm is in general very close to optimal. Further, simulations suggest that non-memoryless file lengths can be accurately approximated by the algorithm.
These methods can likely be applied in more general situations of restless multi-armed bandit problems with constraints.

\section*{Appendix A---Comparison of Max-$\lambda$ and Min-$\lambda$}

This appendix shows that different work conserving policies can give different throughput for
the $N$ single-buffer queue problem of Section \ref{subsection:single_buffer}.
Suppose we have two single-buffer queues and one server.  Let $\lambda_1, \lambda_2$ be the arrival rates of the i.i.d. Bernoulli arrival processes for queues 1 and 2. Assume $\lambda_1 \neq \lambda_2$.
There are 4 system states: $(0,0),~(0,1),~(1,0),~(1,1)$, where state $(i,j)$ means queue 1 has $i$ packets and queue 2 has $j$ packets.  Consider the (work conserving) policy of giving queue 1 strict priority over queue 2.  This is equivalent to the Max-$\lambda$ policy when $\lambda_1>\lambda_2$, and is equivalent to the Min-$\lambda$ policy when $\lambda_1 < \lambda_2$.  Let $\theta(\lambda_1, \lambda_2)$ be the steady state throughput.
Then:
\[ \theta(\lambda_1, \lambda_2) = p_{1,0} + p_{0,1} + p_{1,1}  \]
 where $p_{i,j}$ is the steady state probability of the resulting discrete time Markov chain.  One can solve the global balance equations to show that $\theta(1/2, 1/4) > \theta(1/4, 1/2)$, so that the Max-$\lambda$ policy has a higher throughput than the Min-$\lambda$ policy. In particular, it can be shown that:
 \begin{itemize}
 \item Max-$\lambda$ throughput: $\theta(1/2,1/4) = 0.7$
 \item Min-$\lambda$ throughput: $\theta(1/4, 1/2) \approx 0.6786$
 \end{itemize}


\section*{Appendix B---Proof of Lemma \ref{lem:new}}
This section proves that:
\begin{equation}\label{temp_order-appendix}
\sum_{n=1}^j\tilde{F}_n^\pi(\tau)\leq\sum_{n=1}^j\tilde{F}^\lambda_{n}(\tau)~~\forall j\in\{1,2,\cdots,N\}.
\end{equation}

The case $j=N$ is already established from \eqref{FpileqFlambda}.
Fix $j \in \{1, 2, \ldots, N-1\}$.  Since $\pi$ cannot transmit more packets than Max-$\lambda$ during slot $\tau$,
inequality \eqref{temp_order-appendix} is proved by considering two cases:
\begin{enumerate}
  \item Policy $\pi$ transmits less packets than policy Max-$\lambda$. Then $\pi$ transmits less than $M$ packets during slot $\tau$.  The work-conserving nature of $\pi$  implies all non-empty queues were served, so $\tilde{F}^\pi_n(\tau)=0$ for all $n$ and \eqref{temp_order-appendix} holds.
  \item Policy $\pi$ transmits the same number of packets as policy Max-$\lambda$. In this case, consider the temporary buffer states of the last $N-j$ queues under policy Max-$\lambda$. If  $\sum_{n=j+1}^N\tilde{F}^\lambda_{n}(\tau)=0$, then clearly the following holds
      \begin{equation}\label{reverse_order}
      \sum_{n=j+1}^{N}\tilde{F}^\pi_{n}(\tau)\geq\sum_{n=j+1}^N\tilde{F}^\lambda_{n}(\tau).
      \end{equation}
      Subtracting \eqref{reverse_order} from \eqref{FpileqFlambda} immediately gives \eqref{temp_order-appendix}.
      If  $\sum_{n=j+1}^N\tilde{F}^\lambda_{n}(\tau)>0$, then all $M$ servers of the Max-$\lambda$ system were devoted to serving the largest $\lambda_n$ queues.  So only packets in the last $N-j$ queues could be transmitted by Max-$\lambda$
      during the slot $\tau$.  In particular, $\alpha_n^\lambda(\tau)=0$ for all $n \in \{1, \ldots, j\}$, and so (by \eqref{eq:f-lambda-dude}):
      \begin{equation} \label{eq:last-dude}
      \sum_{n=1}^j\tilde{F}^\lambda_{n}(\tau)=\sum_{n=1}^jF^\lambda_{n}(\tau)
      \end{equation}
      Thus:
      \begin{align}
      \sum_{n=1}^j\tilde{F}^\pi_{n}(\tau)&\leq\sum_{n=1}^jF_n^\pi(\tau) \label{eq:thus-dude1} \\
      &\leq\sum_{n=1}^jF^\lambda_{n}(\tau)\label{eq:thus-dude2} \\
      &=\sum_{n=1}^{j}\tilde{F}^\lambda_{n}(\tau), \label{eq:thus-dude3}
      \end{align}
      where \eqref{eq:thus-dude1} holds by \eqref{eq:f-pi-dude}, \eqref{eq:thus-dude2} holds because
      \eqref{partial_order} is true on slot $t=\tau$, and the last equality holds by \eqref{eq:last-dude}. This proves \eqref{temp_order-appendix}.
\end{enumerate}
\bibliographystyle{unsrt}
\bibliography{bibliography/refs}
%


\end{document}